\documentclass[conference,letterpaper]{IEEEtran}

\usepackage[utf8]{inputenc} 
\usepackage[T1]{fontenc}
\usepackage{url}
\usepackage{ifthen}
\usepackage{cite}
\usepackage[cmex10]{amsmath}% Use the [cmex10] option to ensure complicance
\usepackage{color}
\usepackage{amssymb}
\usepackage{amsthm}
\usepackage{bm}
\usepackage{bbm}
\usepackage{mathrsfs}
\usepackage{mathtools}
\usepackage[ruled]{algorithm2e}
\SetKwProg{Init}{Initialize:}{}{}
\SetKwInOut{Output}{Output}

\DeclarePairedDelimiter{\abs}{\lvert}{\rvert}
\DeclarePairedDelimiter{\norm}{\|}{\|}

\DeclarePairedDelimiter{\bra}{\lbrack}{\rbrack}
\DeclarePairedDelimiter{\bre}{\lbrace}{\rbrace}
\DeclarePairedDelimiter{\para}{(}{)}
\DeclarePairedDelimiter{\dotp}{\langle}{\rangle}

\DeclareMathOperator{\E}{\mathsf{E}}

\newcommand{\R}[0]{\mathbb{R}}

\newtheorem{lemma}{Lemma}

\newtheorem{remark}{Remark}
\newtheorem{assumption}{Assumption}

\begin{document}
\title{Distributed and Rate-Adaptive Feature Compression} 

% %%% Single author, or several authors with same affiliation:
% \author{%
%  \IEEEauthorblockN{Andrew R.~Barron}
%  \IEEEauthorblockA{Department of Statistics and Data Science\\
%                    Yale University\\
%                    New Haven, CT, USA\\
%                    Email: andrew.barron@yale.edu}
% }

%%% Several authors with up to three affiliations:
\author{%
  \IEEEauthorblockN{Aditya Deshmukh}
  \IEEEauthorblockA{Department of Electrical\\ and Computer Engineering\\
                    University of Illinois Urbana-Champaign\\
                    Champaign, IL, USA\\
                    Email: ad11@illinois.edu}
  \and
  \IEEEauthorblockN{Venugopal V. Veeravalli}
  \IEEEauthorblockA{Department of Electrical\\ and Computer Engineering\\
                    University of Illinois Urbana-Champaign\\
                    Champaign, IL, USA\\
                    Email: vvv@illinois.edu}
\and
  \IEEEauthorblockN{Gunjan Verma}
  \IEEEauthorblockA{Army Research Laboratory\\ 
                    Adelphi, MD, USA\\
                    Email: gunjan.verma.civ@army.mil}
}

%%% Many authors with many affiliations:
% \author{%
%   \IEEEauthorblockN{Andrew R.~Barron\IEEEauthorrefmark{1},
%                     Claude E.~Shannon\IEEEauthorrefmark{2},
%                     David Slepian\IEEEauthorrefmark{2},
%                     and Jacob Ziv\IEEEauthorrefmark{2}\IEEEauthorrefmark{3}}
%   \IEEEauthorblockA{\IEEEauthorrefmark{1}%
%                    Department of Statistics and Data Science, Yale University, New Haven, CT, USA,
%                     andrew.barron@yale.edu}
%   \IEEEauthorblockA{\IEEEauthorrefmark{2}%
%                     Bell Telephone Laboratories, Inc.,
%                     Murray Hill, NJ, USA,
%                     \{csh,dsl,jz\}@bell-labs.com}
%   \IEEEauthorblockA{\IEEEauthorrefmark{3}%
%                     Department of Electrical Engineering, Technion---Institute of Technology, Haifa, Israel,
%                     jz@ee.technion.ac.il}
% }

\maketitle

%%%%%%
%% Abstract: 
%% If your paper is eligible for the student paper award, please add
%% the comment "THIS PAPER IS ELIGIBLE FOR THE STUDENT PAPER
%% AWARD." as a first line in the abstract. 
%% For the final version of the accepted paper, please do not forget
%% to remove this comment!
%%

\begin{abstract}
We study the problem of distributed and rate-adaptive feature compression for linear regression. 
A set of distributed sensors collect disjoint features of regressor data. A fusion center is assumed to contain a pretrained linear regression model,
trained on a dataset of the entire uncompressed data. At inference time, the sensors compress their observations and send them to the fusion center through communication-constrained channels, whose rates can change with time. Our goal is to design a feature compression {scheme} that can adapt to the varying communication constraints, while maximizing the inference performance at the fusion center. We first obtain the form of optimal quantizers assuming knowledge of underlying regressor data distribution.
Under a practically reasonable approximation, we then propose a distributed compression scheme
which works by quantizing a one-dimensional projection of the sensor data. We also propose a simple adaptive scheme for handling changes in communication constraints.  
We demonstrate the effectiveness of the distributed adaptive compression scheme through simulated experiments.
\end{abstract}

\section{Introduction}
A prevalent way in which machine learning models are trained involves collecting data from various relevant sources, and training the models on the aggregated data. However, in many applications, the input data is often collected from distributed sources at inference time. Examples include, the Internet of Things (IoT) networks, security systems with surveillance sensors, and driverless cars collecting data from sensors and receiving data from wireless receivers. In these applications, the volume of data is generally high and decisions are time-sensitive, and so it is important to have low latency. Moreover, when the data is being communicated through wireless channels, bit-rates can be quite low  either for energy conservation purposes or because of poor channel conditions. Thus, it is imperative to optimize the data-stream pipelines in order to provide maximum information relevant to the performance of the downstream task. Moreover, in practice these pipelines are also subject to changes in bit-rates, and so it is necessary for the solutions to be adaptive to these changes. In this work, we try to answer the following question:

\quad\textit{How to maximize information (relevant to the downstream task) received at a pretrained model at inference time when input data is collected in a distributed way through communication-constrained channels that are subject to change?}

We consider a distributed sensor network that consists of $m$ sensors and a fusion center. The sensors collect multi-modal observations, compress and quantize them, and send them to the fusion center through communication-constrained channels. The fusion center contains a learning model, which is pretrained on a training dataset of the uncompressed multi-modal data. We consider the goal of designing efficient feature compressors that can adapt to dynamic communication constraints, while maximizing inference performance at the fusion center. In order to design the compressors, we assume that we have access to a \emph{calibration} dataset, which may  be a subset of the dataset on which pretrained model is trained.

There is considerable literature on distributed compression for detection and estimation when the underlying sensor data distributions (or families of distributions) are known (see, e.g., ~\cite{Tsitsiklis93,luo2005universal, zhang2013information}). In this work, we do not assume knowledge of the sensor data distributions. 
% \vvv{What about the literature on distributed estimation?}

% Recently, there has been paradigm-shift from conventional communication to \textit{goal-oriented communication} (see, e.g.,~\cite{zhu2020toward,strinati20216g}), particularly when the goal is inference through machine learning models. Many of these works do not assume knowledge of underlying data distributions and propose data-driven techniques. \vvv{I have no idea what goal-oriented communication is and why it is relevant}

Recently, a line of works have studied the problem of designing distributed compression/quantization schemes for machine learning. A customized quantization scheme for diagonal linear discriminant analysis in a distributed sensor setting was proposed by ~\cite{du2016novel}.~\cite{hanna2020distributed} showed that the problem of designing optimal distributed feature compression schemes is NP-hard and proposed deep neural-network (DNN) based solutions for the task of classification.~\cite{shao2022task} propose a task-relevant feature extraction framework based on the information bottleneck principle, and adopt the distributed information bottleneck framework to formalize a single-letter characterization of the optimal rate-relevance trade-off for distributed feature encoding.~\cite{shao2022task} also proposed DNN-based solutions for their proposed framework. While these works demonstrate that inference performance can be maintained under high levels of compression,  the proposed schemes are not adaptive, i.e., if the communication constraints change, the proposed compression schemes need to be re-trained, which may be impractical for delay-sensitive applications.

Our contributions are as follows:
\begin{enumerate}

    \item We propose a framework for designing optimal compression schemes when the pretrained model is a linear regressor, which also extends to the general learning model case. To the best of our knowledge, this is the first work which analyzes distributed compression schemes for distributed linear regression. Assuming knowledge of the underlying data distribution, we characterize the structure of the optimal compression scheme, and address the difficulty in formulating its empirical version when the underlying data distribution is not known.
    \item Under mild assumptions, we obtain optimal compressors when the pretrained model is a linear regressor. We show that the optimal compressor works by quantizing a one-dimensional projection of the sensor data.
    \item We propose a simple adaptive scheme for handling changes in communication constraints.
    \item Based on the structure of the compression scheme for the linear regression model, we propose an adaptive compression scheme based on Vector-Quantized Variational Auto-Encoders (VQ-VAEs) (see~\cite{van2017neural}) for the case where there is a general pretrained model at the fusion center. Our proposed scheme is motivated by the fact that VQ-VAE based compression works by projecting the data onto a low-dimensional space. We further quantize the latent representations, which matches the compression scheme obtained for pretrained linear regressors.
    \item We show that the adaptive scheme renders re-training the VQ-VAEs unnecessary when the communication constraints change.
    \item We demonstrate the effectiveness of the distributed adaptive compression scheme through experiments.
\end{enumerate}

\begin{figure*}[t]
\centering
\includegraphics[width=0.75\linewidth]{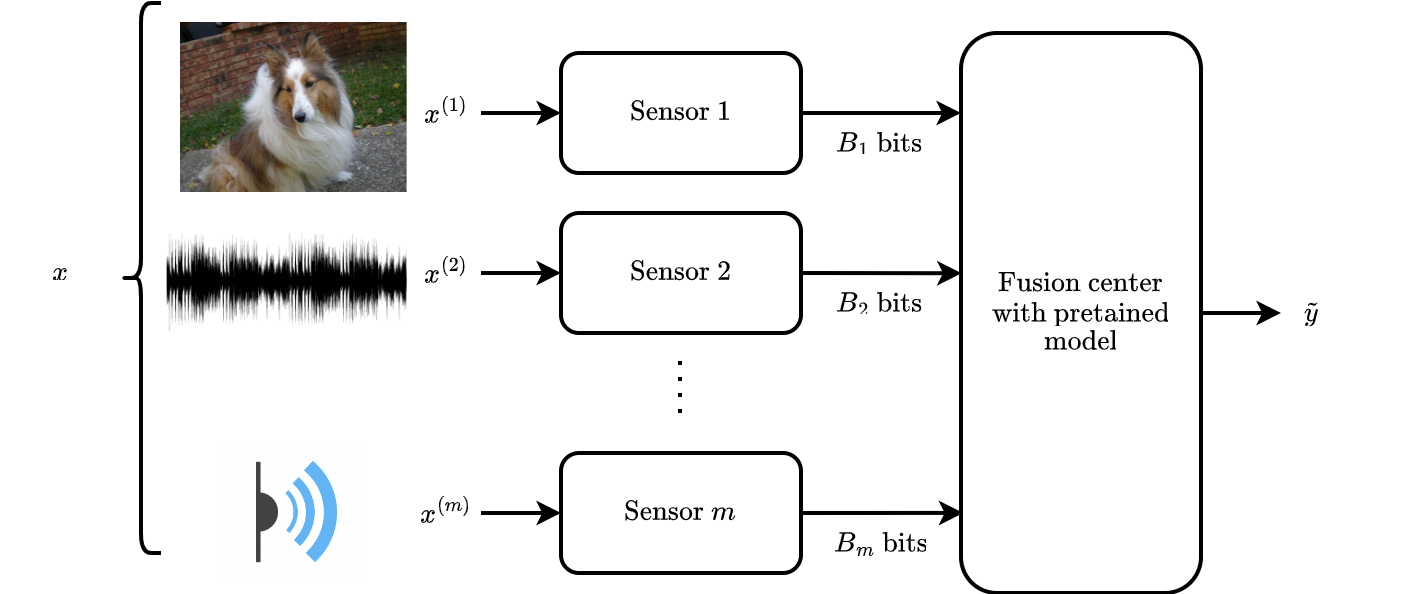}
\caption{Distributed system with multi-modality}
\end{figure*}

% \begin{figure*}[t]
% \centering
% \includegraphics[width=0.75\linewidth]{figs/plain_net.drawio.png}
% \caption{Distributed system with multi-modality}
% \end{figure*}

\section{Problem Setup}

We first consider the case where the pretrained model at the fusion center is a linear regressor.

Consider the following linear model:
\begin{align}
    \hat{y} = \dotp{x,\beta}
\end{align}
where $\hat{y}\in\R$ is the response variable, $x\in\R^d$ is the regressor, and $\beta\in\R^d$ is the learned parameter.

The sensors are assumed to observe different features of the regressor $x$. Let the $i$th sensor observe features indexed by $\mathcal{S}_i\subset [d]$. We assume that $\mathcal{S}_i$'s are non-empty and pairwise disjoint. Let $d_i=\abs{\mathcal{S}_i}$. We denote the observation of $i$th sensor by $x^{(i)}\in\R^{d_i}$:
\begin{align}
    x^{(i)} = [x_{r}, r\in \mathcal{S}_i].
\end{align}
Assume that the $i$th sensor quantizes its observation using $B_i$ bits (equivalently, $K_i=2^{B_i}$ levels), by using a quantizer of the form $Q_i:\R^{d_i}\to\bre*{c_1,\dots,c_{K_i}}$, where $c_k$'s are the quantization vectors (to be optimized) lying in $\R^{d_i}$. Let the quantized observation of the $i$th sensor be denoted as $\tilde{x}^{(i)}\in\R^{d_i}$:
\begin{align}
    \tilde{x}^{(i)}=Q_i(x^{(i)}).
\end{align}
Let the complete quantization output of the regressor $x$ through this distributed quantization scheme be denoted as $\tilde{x}$.

% \vvv{[There is a problem with the notation here. Above you use $x_r$ to denote the $r$-th component of the vector $x$, whereas here you use it to index different vectors.]} 
With slight abuse of notation, let the calibration data be denoted as $\bre*{x_j}_{j=1}^n$ (the individual components of $x_j$ will be denoted by $x_{j,r}$). Let $\hat{y}_j=\dotp{x_j, \beta}$ and $\tilde{y}_j=\dotp{\tilde{x}_j, \beta}$ be the estimated response for the unquantized and quantized regressors $x_j$ and $\tilde{x}_j$, respectively. We represent $n$ such equations in the following matrix form:
\begin{align}
    \hat{Y} = X\beta \quad\text{ and }\quad \tilde{Y} = \tilde{X}\beta
\end{align}
where $j$th row of the matrices $X$ and $\tilde{X}$ are $x_j$ and $\tilde{x}_j$, respectively. For analysis purposes, let $X^{(i)}$ and $\tilde{X}^{(i)}$ be the matrices with $j$th row as $x^{(i)}_j$ and ${\tilde{x}_j^{(i)}}$, respectively, $j=1,\dots,n$. Let the corresponding \textit{contributions} of $i$th sensor be denoted by
\begin{align}
    \hat{Y}^{(i)} = X^{(i)}\beta^{(i)} \quad\text{ and }\quad \tilde{Y}^{(i)} = \tilde{X}^{(i)}\beta^{(i)}
\end{align}
respectively.
Note that we denote random variables with a boldface font. A reasonable objective is to design the quantizers $Q_1,\dots,Q_m$ in order to minimize the MSE:
\begin{align}\label{eq:obj_linreg}.
    \min_{Q_1,\dots,Q_m} \E\bra*{(\tilde{\bm y}-\hat{\bm y})^2}
\end{align}
Here, the expectation is taken over the distribution of the input $x$, conditioned on the learned parameter $\beta$ being fixed. However, the underlying data distribution is unknown in most practical cases. Conventionally, the method of empirical risk minimization (ERM) is applied to approximately solve an objective of the form in \eqref{eq:obj_linreg}. However, as we argue in the following remark, the ERM approach does not provide a complete solution for the problem considered.

\begin{remark}The empirical form of the optimization problem in  \eqref{eq:obj_linreg} is
\begin{align}\label{eq:obj}
    &\min_{Q_1,\dots,Q_m} \frac{1}{n}\norm*{\tilde{Y}-\hat{Y}}_2^2\\
    =&\min_{Q_1,\dots,Q_m} \frac{1}{n}\norm*{\sum\limits_{i=1}^m\para*{\tilde{Y}^{(i)}-\hat{Y}^{(i)}}}_2^2.
\end{align}

Although it is possible to obtain a local minimum of this empirical objective through best-response dynamics by considering a potential game (see \cite{monderer1996potential}), the solution does not capture the entire form of the quantizers $Q_i(.)$, but only the provides us with the quantizers' outputs on the calibration data, i.e., $Q_i(x^{(i)}_j)$. Particularly, the formulation in \eqref{eq:obj} does not capture the constraint that at inference time, the $i$th quantizer has access to data only at the $i$th sensor.
\end{remark}

Hence, we first analyze the population-based objective given in \eqref{eq:obj_linreg}, and obtain the structure of the optimal quantizers in the following lemma.
%, whose proof is provided in the Appendix.

\begin{lemma}\label{lemma:1}
    The optimal quantizers for optimization problem in \eqref{eq:obj_linreg} simultaneously satisfy the following structure (which is provided for an arbitrary $i$th sensor): 
    \begin{align}
    Q_i(x^{(i)})=\bar{c}_k^{(i)},\;\; \text{ if } x^{(i)}\in \bar{\mathcal{D}}_k^{(i)}, \text{ for }k\in [K_i],
    \end{align}
    where $\bar{c}_k^{(i)}$ and $\bar{\mathcal{D}}_k^{(i)}$ are the minimizers of
    \begin{align}\label{eq:min_cD}    \min_{c^{(i)}_k,\mathcal{D}^{(i)}_k}\sum\limits_{k=1}^{K_i}\int\limits_{\mathcal{D}^{(i)}_k} f_i(c^{(i)}_k,x^{(i)})p_i(x^{(i)})d\mu_i(x^{(i)}),
    \end{align}
    where  $\mu_i$ is the Lebesgue measure on $\R^{d_i}$, $p_i$ is the probability density function of $\bm x^{(i)}$ with respect to $\mu_i$ and
\begin{align}\label{eq:lin_reg}
f_i(c^{(i)}_k,x^{(i)})=&\E\Bigg[\Bigg(\dotp{c^{(i)}_k-x^{(i)},{\beta}^{(i)}}+\\
    &\sum\limits_{i'\neq i}\dotp{\tilde{\bm x}^{(i')}-\bm x^{(i')},{\beta}^{(i')}}\Bigg)^2\Bigg|\bm x^{(i)}=x^{(i)}\Bigg]\nonumber.
\end{align}
    Note that $\tilde{\bm x}^{(i')}$s (where $i'\neq i$) are obtained similarly through optimal quantizers for sensors other than the $i$th sensor, i.e., $\tilde{\bm x}^{(i')}=Q_{i'}(\bm x^{(i')})$.
\end{lemma}

\begin{proof} 
To obtain the structure of the optimal quantizers, we focus our attention on an arbitrary $i$th quantizer. Note that the quantizer $Q_i$ has the following general form:
\begin{align}
    Q_i(x^{(i)}) = c^{(i)}_{k}, \text{ if }x^{(i)}\in\mathcal{D}_k^{(i)}.
\end{align}
Fix quantizers $Q_1,\dots,Q_{i-1},Q_{i+1},\dots,Q_m$ to be the optimal quantizers for sensors other than $i$th sensor and consider the following objective:
\begin{align}
    &\min_{Q_i} \E\bra*{(\tilde{\bm y}-\hat{\bm y})^2}\\
    =&\min_{Q_i} \E\bra*{\E\bra*{(\tilde{\bm y}-\hat{\bm y})^2|\bm x^{(i)}}}\\
    =& \min_{c^{(i)}_k,\mathcal{D}^{(i)}_k}\sum\limits_{k=1}^{K_i}\int\limits_{\mathcal{D}^{(i)}_k}g_i(x^{(i)})p_i(x^{(i)})d\mu_i(x^{(i)})\label{eq:gp}
\end{align}
where
\begin{align}
    &g_i(x^{(i)})\\
    =&\mathsf{E}\bra*{(\tilde{\bm y}-\hat{\bm y})^2|\bm x^{(i)}=x^{(i)}}\\
    =&\mathsf{E}\bra*{\para*{\sum\limits_{j=1}^m\dotp{\tilde{\bm x}^{(j)}-\bm x^{(j)},\beta^{(j)}}}^2\Bigg|\bm x^{(i)}=x^{(i)}}\\
    =&\E\Bigg[\Bigg(\dotp{Q_i(x^{(i)})-x^{(i)},{\beta}^{(i)}}+\label{eq:end_gp}\\
    &\sum\limits_{i'\neq i}\dotp{\tilde{\bm x}^{(i')}-\bm x^{(i')},{\beta}^{(i')}}\Bigg)^2\Bigg|\bm x^{(i)}=x^{(i)}\Bigg]\nonumber
\end{align}
Note that $g_i(x^{(i)})=f_i(c^{(i)}_k,x^{(i)})$ when $x^{(i)}\in\mathcal{D}_k^{(i)}$, where $f_i(c^{(i)}_k,x^{(i)})$ is as defined in \eqref{eq:lin_reg}. Hence, \eqref{eq:gp}--\eqref{eq:end_gp} can be written as \eqref{eq:min_cD}--\eqref{eq:lin_reg}.
\end{proof}

\begin{remark}
Assuming knowledge of the underlying distribution of $\bm x$, it is possible to find a local minimum of the optimization problem in \eqref{eq:obj_linreg} through best-response dynamics, which involves solving \eqref{eq:min_cD} by applying the generalized Linde-Buzo-Gray (LBG)algorithm (see~\cite{linde1980algorithm}). However, it is not so straightforward to solve \eqref{eq:min_cD} empirically when the distribution of $\bm x$ is not known.
\end{remark}

We consider a practically reasonable approximation that the optimal quantization error at sensor $i'\neq i$ has zero mean conditioned on the observation at sensor $i$.
\begin{assumption}
    \begin{align}
        \mathsf{E}\bra*{\sum\limits_{i'\neq i}\tilde{\bm x}^{(i')}-\bm x^{(i')}\Big|\bm x^{(i)}=x^{(i)}}\approx 0.
    \end{align}
\end{assumption}

Under Assumption 1, we can express $f_i(c^{(i)}_k,x^{(i)})$ as 
\begin{align}
    &f_i(c^{(i)}_k,x^{(i)})=\dotp{c^{(i)}_k-x^{(i)},{\beta}^{(i)}}^2+\nonumber\\
    &2\dotp{c^{(i)}_k-x^{(i)},{\beta}^{(i)}}\mathsf{E}\bra*{\sum\limits_{i'\neq i}\dotp{\tilde{\bm x}^{(i')}-\bm x^{(i')},{\beta}^{(i')}}\Bigg|\bm x^{(i)}=x^{(i)}}\nonumber\\
    &+\mathsf{E}\bra*{\para*{\sum\limits_{i'\neq i}\dotp{\tilde{\bm x}^{(i')}-\bm x^{(i')},{\beta}^{(i')}}}^2\Bigg|\bm x^{(i)}=x^{(i)}}\label{eq:12}\\
    &\approx\dotp{c^{(i)}_k-x^{(i)},{\beta}^{(i)}}^2+0+r_i(x^{(i)}),
\end{align}
where the third term in \eqref{eq:12} is denoted as $r_i(x^{(i)})$.  The second term in \eqref{eq:12} is taken to be 0 due to the approximation considered. The minimizers in \eqref{eq:min_cD} are now equivalent to the minimizers in
 \begin{align}\label{eq:min_cD_2}    \mathop{\arg\min}_{c^{(i)}_k,\mathcal{D}^{(i)}_k}\sum\limits_{k=1}^{K_i}\int\limits_{\mathcal{D}^{(i)}_k} f'_i(c^{(i)}_k,x^{(i)})p(x^{(i)})d\mu(x^{(i)})
    \end{align}
    where
\begin{align}
f'_i(c^{(i)}_k,x^{(i)})&=\dotp{c^{(i)}_k-x^{(i)},{\beta}^{(i)}}^2\\
    &=\para*{\dotp{c^{(i)}_k,{\beta}^{(i)}}-\dotp{x^{(i)},{\beta}^{(i)}}}^2.
\end{align}
Observe that the conditional expectation in \eqref{eq:lin_reg} makes formulating a empirical version of \eqref{eq:min_cD} difficult. However, the objective in \eqref{eq:min_cD_2} is devoid of the conditioning, and hence we can formulate an empirical version of it through a K-Means clustering method with following assignment objective:
\begin{align}
\mathop{\arg\min}_{a^{(i)}}\sum\limits_{k=1}^{K_i}\sum\limits_{j:a^{(i)}_j=k}\para*{\dotp{c^{(i)}_k,{\beta}^{(i)}}-\dotp{x_j^{(i)},{\beta}^{(i)}}}^2
\end{align}
where
\begin{align}
c_k^{(i)} = \frac{1}{|j:a^{(i)}_j=k|}\sum\limits_{j:a^{(i)}_j=k} x^{(i)}_j.
\end{align}
Note that we can express the above K-Means objective as a one-dimensional K-Means objective with the projected data $\hat{y}^{(i)}$:
\begin{align}
\mathop{\arg\min}_{a^{(i)}}\sum\limits_{k=1}^{K_i}\sum\limits_{j:a^{(i)}_j=k}\para*{h^{(i)}_k-\hat{y}_j^{(i)}}^2
\end{align}
where
\begin{align}
h_k^{(i)} = \frac{1}{|j:a^{(i)}_j=k|}\sum\limits_{j:a^{(i)}_j=k} \hat{y}^{(i)}_j.
\end{align}

 This suggests a simple algorithm (Algorithm \ref{alg:lin_reg}) to solve the empirical form of \eqref{eq:obj_linreg}: 1) for each sensor, run a one-dimensional K-means algorithm (see~\cite{gronlund2017fast}) on the projected calibration data to obtain cluster assignments, 2) obtain cluster centers of sensor calibration data using the cluster assignments, 3) set quantizers which output the nearest cluster center based on the projected sensor data at inference time.

\begin{algorithm}[ht]\label{alg:lin_reg}
\caption{Distributed quantization for linear regression}
\KwData{Calibration dataset $\{\bm x_j\}_{j=1}^n$, learned parameter $\beta$, number of clusters $\{K_i\}_{i=1}^m$.}
\Init{\upshape Compute $\hat{Y}^{(i)}=X^{(i)}\beta^{(i)}$.}{}
\For{$i=1$ to $m$}{
$\bm a^{(i)}=$ 1D-K-Means$\para*{\hat{Y}^{(i)},K_i}$
}
Set $\bm{c}_k^{(i)} = \frac{1}{\abs{j:a^{(i)}_j=k}} \sum\limits_{j:a^{(i)}_j=k} x_j^{(i)},\, k\in[K_i], i\in [m]$.\\
Set $\mathcal{D}^{(i)}_k$ as the set:\\ $\bre*{
x^{(i)}: \dotp{c^{(i)}_k-x^{(i)},{\beta}^{(i)}}^2\leq \dotp{c^{(i)}_{k'}-x^{(i)},{\beta}^{(i)}}^2, \forall k'\neq k
}$\\
% \Output{Cluster representatives and cluster regions for $i$th sensor: $\{c^{(i)}_k\}_{k=1}^{K_i}$, $\{D^{(i)}_k\}_{k=1}^{K_i}$ for every $i\in[m]$.}
\Output{Quantizers: $Q_i(x^{(i)})=c_k^{(i)}$, if $x^{(i)}\in\mathcal{D}_k^{(i)}$.}
\end{algorithm}

% Observe that the algorithm suggests a simpler form of quantizer: project data onto $\beta^{(i)}$ and obtain $z_j^{(i)}=\dotp{\bm x_j^{(i)}, \beta^{(i)}}$, use the cluster centers obtained through 1D-Kmeans on the projected calibration data for nearest neighbor quantization and at inference time send the quantized projected data: $\tilde{z}_j$. The fusion center then outputs the response $\tilde{y}=\sum\limits_{i=1}^m\tilde{z}_j^{(i)}$.

\subsection{Adaptive Quantization}\label{sec:quant}
We consider the scenario wherein the communication constraints between the sensors and the fusion center are susceptible to change. In this case, we propose Algorithm \ref{alg:adap} to cluster the cluster representatives obtained by the quantizers using Weighted-K-Means algorithm to reduce the number of cluster representatives, and thereby reduce the number of bits used. 

\begin{algorithm}\label{alg:adap}
\caption{Adaptive K-Means}
\KwData{Cluster centers $\mathcal{C}=\{c_k\}_{k=1}^K$, number of datapoints assigned to each cluster $\mathcal{W}=\{n_k\}_{k=1}^K$ and new number of clusters $K'\leq K$.} 
Obtain $\bm a$ = Weighted-K-Means$(\mathcal{C},\mathcal{W},K')$.\\
Set $\bm{c'}_{i} = \frac{1}{\sum\limits_{k:a_k=i}n_k} \sum\limits_{k:a_k=i} n_kc_k, \text{ where } i\in[K'].$\\
\Output{New cluster centers $\{c'_i\}_{i=1}^{K'}$.}
\end{algorithm}

\subsection{Experimental Results}
We generated a synthetic calibration dataset as follows: we generate $n=10000$ datapoints in $\R^{100}$ from multivariate Gaussian with a mean and a covariance matrix that were randomly generated and fixed. The parameter $\beta\in\R^{100}$ is fixed to a realization of a random vector with independent normally distributed entries. The responses of the model on the calibration data were computed as $\hat{y}=X\beta$. We considered 10 sensors with each sensor observing 10 features. For various choices of bits/sensor we ran  Algorithm~\ref{alg:lin_reg} to obtain the proposed (non-adaptive) quantizers. For the same choices of bits/sensor, we also obtain adaptive quantizers, in which new cluster centers were obtained for fewer than 10 bits, by running Algorithm \ref{alg:adap} with the cluster centers corresponding to the highest number of bits, i.e., 10 obtained by  Algorithm~\ref{alg:lin_reg}. We also compared the performance of the non-adaptive and adaptive strategies with the following strategy for quantization at the sensors, which is agnostic to the pretrained model at the fusion center: for $i$th sensor, use K-Means clustering with the sensor calibration data $X^{(i)}$ and $K_i$ number of clusters. Figure \ref{fig:lin_reg} shows the MSE values obtained by the non-adaptive strategy, the proposed adaptive strategy, and the naive strategy.  The MSE values were computed by Monte-Carlo simulation with 100,000 test datapoints. Number of bits assigned to each sensor was varied form 1 to 10. We make the following observations: 1) the proposed adaptive strategy  performs as well as the non-adaptive strategy, and 2) the proposed adaptive strategy had an order of magnitude improvement in MSE performance over the naive strategy for moderate number of bits. We observed similar trends across various synthetically generated models.

% \begin{table}[h!]
% \centering
% \begin{tabular}{ |c|c|c|c| } 
% \hline
% \#bits/sensor & Proposed strategy & Naive strategy \\
% \hline
% 1 & 869.87 & 869.87 \\ 
% 2 & 276.75 & 329.26 \\ 
% 3 & 73.48 & 143.97 \\ 
% 4 & 26.26 & 93.40 \\ 
% 5 & 11.39 & 71.29 \\ 
% \hline
% \end{tabular}
% \caption{MSE values obtained by proposed strategy vs naive strategy.}
% \label{table:1}
% \end{table}
\begin{figure}[ht]
\centering
\includegraphics[width=\linewidth]{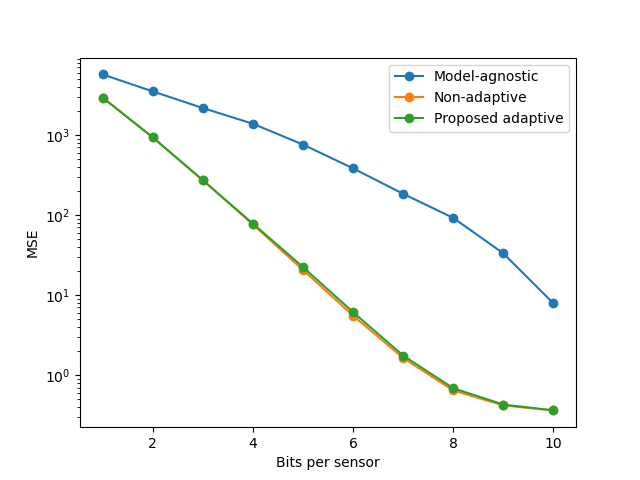}
\caption{MSE values obtained by the model-agnostic, non-adaptive, and proposed adaptive strategies.  \label{fig:lin_reg}} 

\end{figure}

\section{General Learning Model}
\begin{figure*}[t]
\centering
\includegraphics[width=0.75\linewidth]{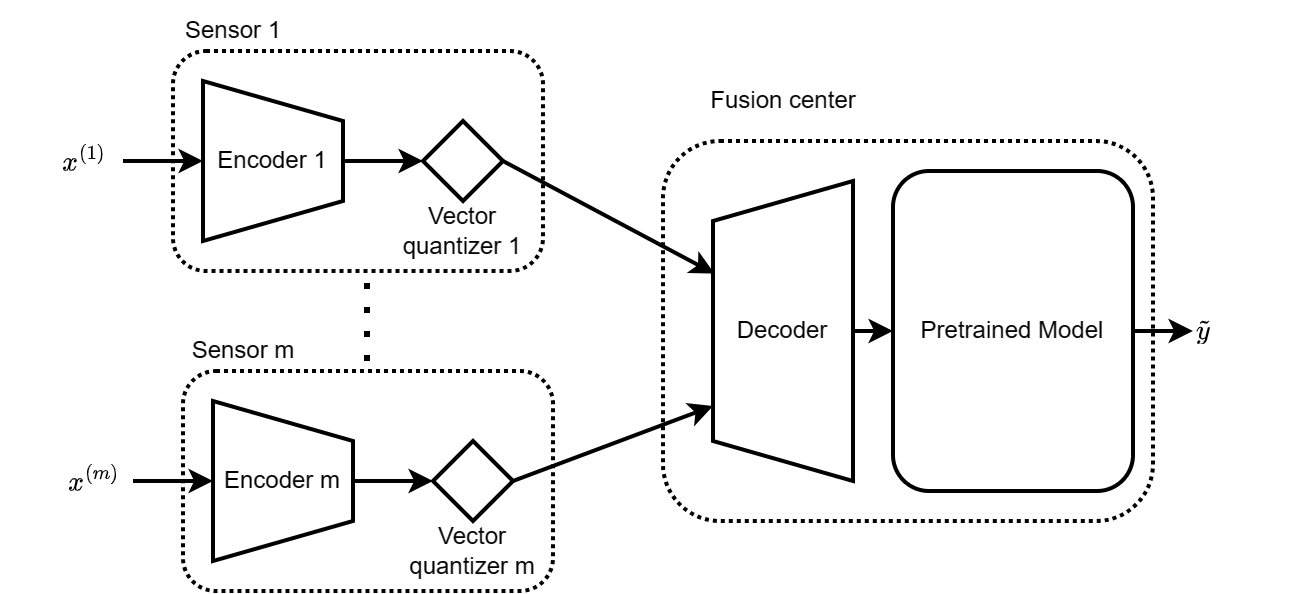}
\caption{VQ-VAE based distributed system}
\label{fig:vqvae}
\end{figure*}
We now consider a general pretrained learning model at the fusion center:
\begin{align}
    \hat{y} = f(x;\theta)
\end{align}
where $x\in\R^{d_\mathrm{in}}$ is the input, $\hat{y}\in\R^{d_\mathrm{out}}$ is the output, and $\theta$ denotes the learned parameters of the model. Let the loss function minimized during training be denoted by $\ell(\hat{y},y)$, where $y$ denotes the ground-truth label in the training dataset. It is assumed that the loss function $\ell$ is differentiable with respect to the input $x$. Following the conclusion of the previous section, we posit that the sensor observations should be projected (or encoded) first in a low-dimensional space and then quantized. The fusion center decodes the received quantized encoded observations and then applies the learned model to produce the output.

Note that the problem of designing optimal encoders and decoder is quite difficult for the general learning model. We propose using VQ-VAEs \cite{van2017neural} for this purpose. VQ-VAEs have shown great empirical success in compressing images and audios (for the purpose of reconstruction) in an unsupervised fashion, as well as being good generative models for generating images and audio signals. In this work, we demonstrate the effectiveness of VQ-VAEs in quantizing sensor inputs in the distributed setting, when the downstream task is not reconstruction, but a more general task captured by a pretrained model.

A VQ-VAE architecture comprises of an encoder $E$, a vector-quantizer equipped with a codebook and a decoder $D$. Let $\bm{z}_e(x)=E(x)$ be the latent representations generated by the encoder. Typically, the latent representation takes the form of a tensor in $\R^{h\times w\times d_l}$. We term the following values as \textit{pixels} of the latent representation:
\begin{align}[z_e(x)]_{i,j}\in\R^{d_l}.
\end{align}
The sizes of first two dimensions, $h$ and $w$, are fixed given the input size $d_{\text{in}}$ and the encoder. The last dimension $d_l$ is a hyper-parameter, which can be varied.  Let the codebook be represented by a set of $K$ vectors in $\R^{d_l}$: $\{c_1,\dots,c_K\}$. The vector quantizer acts along the last dimension of $z_e(x)$ and outputs $z_q(x)$:
\begin{align}
    &[z_q(x)]_{i,j} = c_{k_{i,j}},\\ \text{ where }&k_{i,j}\in\mathop{\arg\min}_{t\in [K]} \|[z_e(x)]_{i,j}-c_t\|_2.
\end{align}
Let the output of the decoder be $\tilde{x}=D(x)$. Typically,
VQ-VAEs are trained with the goal of minimizing reconstruction loss.

In the distributed setting, the VQ-VAEs components are split up. The encoders and vector-quantizers are located at the sensors. We use super-script $(i)$ to denote a quantity/component at $i$th sensor. We assume a general form of a decoder $D$ at the fusion center (see Figure \ref{fig:vqvae}). Note that this captures the case wherein there are different decoders corresponding to each sensor at the fusion center. The quantized latent representations $z_q^{(i)}(x^{(i)})$ are communicated to the fusion center using $B_i=\log_2(K_i)$ bits. The fusion center concatenates these quantized latent representations and feeds them into the decoder. The decoder outputs $\tilde{x}$, which is then fed into the pretrained model. Let the calibration data be denoted by $(x_j,y_j)$, $j=1,\dots,n$. The whole distributed compression system is trained end-to-end while fixing the pretrained model. Since our goal is associated with the downstream task, we propose the following loss function:
\begin{align}
    L &= \frac{1}{n}\sum\limits_{j=1}^n\bra*{\ell(\tilde{y}_j, y_j)+\frac{1}{m}\sum\limits_{i=1}^m L_{q,j}^{(i)}}\label{eq:loss}
\end{align}
where
\begin{align}
    L_{q,j}^{(i)} &=\beta_1\|\mathrm{sg}(z^{(i)}_e(x_j^{(i)}))-z^{(i)}_q(x_j^{(i)})\|^2_2\label{eq:L_q}\\
    &\hspace*{0.5in} + \beta_2\|z^{(i)}_e(x_j^{(i)}-\mathrm{sg}(z^{(i)})_q(x_j^{(i)})\|^2_2\nonumber\\
    \tilde{y} &= f(\tilde{x};\theta)\\
    \tilde{x} &= D(z_q(x))\\
    z_q(x) &= \mathrm{concat}(z_q^{(1)}(x^{(1)}),\dots,z_q^{(m)}(x^{(m)}))\label{eq:concat}.
\end{align}
%\vvv{Note that $y_j$ can be the model generated target values from the calibration data, i.e., $y_j=f(x_j;\theta)$, in case of absence of ground-truth target values in the calibration data.}
In \eqref{eq:L_q}, $\mathrm{sg(.)}$ refers to the stop-gradient operation which blocks gradients from flowing into its argument, and the concat function in \eqref{eq:concat} is the concatenation of the components into a single tensor/vector. Note that the loss terms on the RHS of \eqref{eq:L_q} are the codebook learning loss and the commitment loss as described in~\cite{van2017neural}.

\subsection{Adaptive Compression}
The distributed system is initially trained with the highest number of bits per sensor. The encoders and decoder(s) obtained are then fixed, and the codebooks of the vector quantizers are stored at the corresponding sensors and the fusion center. At inference time, when the communication constraints change, new codebooks are obtained in the following adaptive manner, keeping the encoders and decoder(s) fixed. We use the same approach as proposed in Section \ref{sec:quant} to obtained a codebook of reduced size, i.e., we use Algorithm \ref{alg:adap} to cluster the cluster representatives in the stored codebooks using a weighted K-Means method to 
%reduce the number of cluster representatives, and thereby 
reduce the number of bits used. 
%We demonstrate significant preservation of performance when adaptive quantization is employed. 
Experimental results are provided in the next section.

\subsection{Experiments}

% The number of bits used by VQ-VAEs are quite high compared to approaches in works by~\cite{hanna2020distributed} and~\cite{shao2022task}. However, these approaches are non-adaptive and VQ-VAEs facilitate adaptive compression schemes, and so we do not compare the performance against these works. \vvv{I don't think you need to say that we use more bits than the other works. We are not sure about this, especially wrt Hann et al since we never really got their code to work.}
We evaluate the proposed VQ-VAE based distributed compression schemes on two datasets: MNIST Audio+Image and CIFAR-10 with the downstream task of classification. We do not compare the performance of the proposed schemes with a model-agnostic compression scheme, e.g., K-Means clustering on the sensor data, since model-agnostic compression schemes are prohibitively computationally expensive due to the large size of sensor data in practical applications, and are known to perform poorly. The hyper-parameters $\beta_1$ and $\beta_2$ in \eqref{eq:L_q} were set to 0.4 and 0.1 respectively. For the proposed adaptive scheme, we used uniform weights while applying weighted K-Means method. All experiments were run on a single NVIDIA GTX-1080Ti GPU with PyTorch.
\begin{figure}[ht]
\centering
\includegraphics[width=0.75\linewidth]{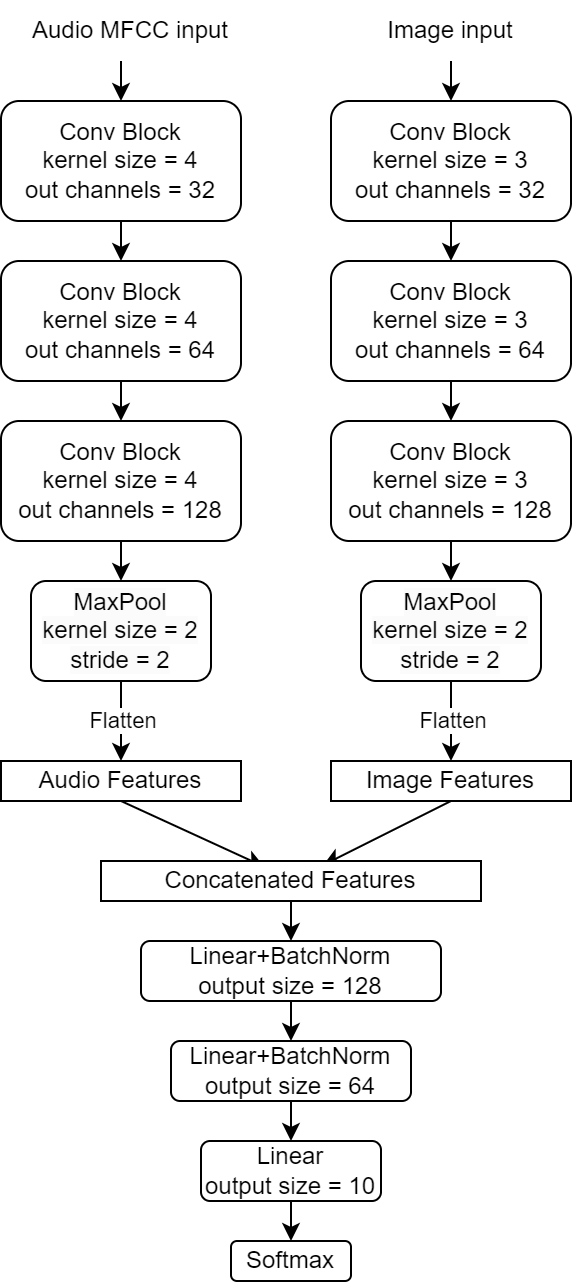}
\caption{DNN architecture for MNIST Audio+Image experiment\label{fig:dnn}}
\end{figure}
\subsubsection{MNIST Audio and Image}
We considered a pretrained DNN (with accuracy 99\%) trained on the MNIST Audio and Image dataset\footnote{DOI 10.5281/zenodo.3515934}. Figure \ref{fig:dnn} shows the architecture of the DNN employed. We considered 2 sensors, one of which senses the image and the other senses the corresponding audio. We considered separated decoders for the sensors at the fusion center. The architectures of the encoders and the decoders were set to be as proposed by~\cite{van2017neural}, with no modifications. The size of the dimension of codebook vectors was fixed to be $d_l^{(i)}=8$ for each sensor. The distributed system was trained using SGD with learning rate $5\times10^{-3}$ for 100 epochs, with a batch size of 256. Figure \ref{fig:mnist} shows the performance of the non-adaptive scheme (VQ-VAEs trained with the exact number of bits specified in the communication constraints) and the proposed adaptive scheme when number of bits/pixel of latent representation per sensor is varied from 1 to 10. We see that the performance is unchanged until number of bits is reduced by 80\%, which clearly illustrates the effectiveness of our adaptive quantization strategy. %\vvv{Need to also say something about the non-adaptive strategy.}

\begin{figure}[ht]
\centering
\includegraphics[width=\linewidth]{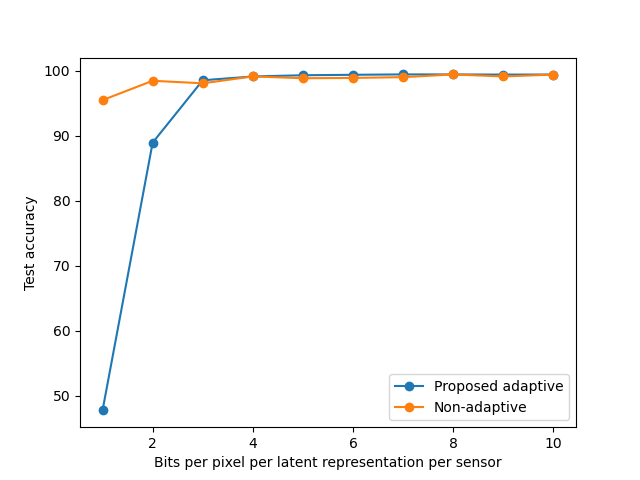}
\caption{Performance of the non-adaptive and  proposed adaptive   strategies on the MNIST Audio+Image dataset. \label{fig:mnist} 
%\vvv{add non-adaptive}
}
\end{figure}

\begin{figure}[ht]
\includegraphics[width=\linewidth]{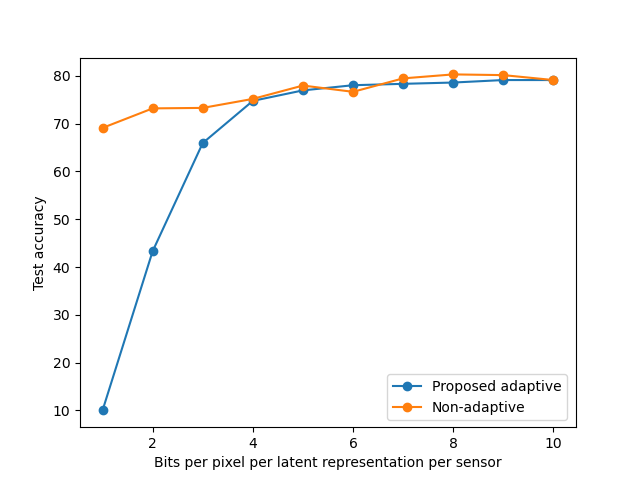}
\caption{Performance of the non-adaptive and  proposed adaptive   strategies on CIFAR 10 dataset.\label{fig:cifar1}%\vvv{remove strategy from legend}
}
\end{figure}

\subsubsection{CIFAR-10}
We considered a pretrained VGG-13 classifier (see~\cite{vgg}) trained on CIFAR-10 dataset (\cite{krizhevsky2009learning}), with accuracy 94\%. We considered 4 sensors, each of which observes a quadrant of an image. We assumed a common decoder at the fusion center. The size of the dimension of codebook vectors was fixed to be $d_l^{(i)}=32$ for each sensor. The architectures of the encoders and the decoder were set to be as proposed by~\cite{van2017neural}, with no modifications. The distributed system was trained using the Adam optimizer (see~\cite{adam}) with learning rate $10^{-4}$ for 1000 epochs, with a batch size of 256. Figure \ref{fig:cifar1} shows the performance of the non-adaptive strategy (trained with the exact number of bits specified in the communication constraints) and the proposed adaptive scheme. The number of bits/pixel of the latent representation per sensor is varied from 1 to 10. We observe that the proposed adaptive scheme performs nearly as well as the non-adaptive scheme until the number of bits is reduced by 70\%.

\section{Conclusion}
We studied the problem of distributed and adaptive feature compression in a distributed sensor network, wherein a set of sensors observe disjoint multi-modal features, compress them, and send them to a fusion center containing a pretrained learning model for inference for a downstream task. To gain insight, we first analyzed the case where the pretrained model is a linear regressor. Assuming knowledge of the underlying data distribution, we obtained structure of the optimal compression scheme. Under a practically reasonable approximation, we leverage the aforementioned structure to develop an algorithm that does not require knowledge of underlying data distribution. We showed that the algorithm works by quantizing a one-dimensional projection of the sensor data. We also proposed an adaptive strategy to handle changes in communication constraints. Experimentally, we demonstrated that the proposed adaptive algorithm significantly outperforms a model-agnostic quantization strategy, in which the sensor observations are quantized without regard to the pre-trained model at the fusion center. For the case when the pretrained model is a general learning model, we proposed a VQ-VAE based compression scheme, which is motivated by the fact that VQ-VAE based compression works by projecting the data onto a low-dimensional space. We further quantized the latent representations, guided by the compression scheme obtained for pretrained linear regressors. We further showed that the adaptive strategy proposed for case of linear regression can also be applied effectively to the VQ-VAE based compression scheme. We demonstrated the effectiveness of the VQ-VAE based distributed and adaptive compression scheme on MNIST Audio+Video and CIFAR10 datasets.

\bibliographystyle{ieeetr}
\bibliography{refs}
\end{document}